\documentclass[12pt,psamsfonts]{amsart}
\usepackage{etex}

\usepackage{euscript,amsmath,amsthm,amsfonts,amssymb}
\usepackage{tikz}
\usepackage{graphicx}
\usepackage[all,arc,knot,web]{xy}
\SelectTips{lu}{12}
\xyoption{arc}
\usepackage{epsfig}
\usepackage[dvips]{pstricks}
\usepackage{subfigure}
\usepackage{color}

\newcommand{\DoubleLineVertex}{
\xy
{\ar@{=}^{k}_{i} (-6,4);(-0.12,0.12)};
{\ar@{=}^{j}_{k} (6,4);(0.12,0.12)};
{\ar@{=}^{i}_{j} (0,-7);(0,-0.2)};
\endxy
}

\newcommand{\DoubleLineHilbertSpace}{
\xy
{\ar@{=} (-6,9);(-0.12,5.12)};
{\ar@{=} (6,9);(0.12,5.12)};
{\ar@{=} (-6,-9);(-0.12,-5.12)};
{\ar@{=} (6,-9);(0.12,-5.12)};
{\ar@{=}^{i}_{j} (0,-4.88);(0,4.88)};
(0,5)*\frm<2pt>{*};
(0,-5)*\frm<2pt>{*}; 
{\ar@{=} (6,9);(11.88,5.12)};
{\ar@{=} (18,9);(12.12,5.12)};
{\ar@{=} (6,-9);(11.88,-5.12)};
{\ar@{=} (18,-9);(12.12,-5.12)};
{\ar@{=}^{k}_{l} (12,-4.88);(12,4.88)};
(12,5)*\frm<2pt>{*};
(12,-5)*\frm<2pt>{*};
(0,-5)*\frm<2pt>{*}; 
{\ar@{=} (18,9);(23.88,5.12)};
{\ar@{=} (30,9);(24.12,5.12)};
{\ar@{=} (18,-9);(23.88,-5.12)};
{\ar@{=} (30,-9);(24.12,-5.12)};
{\ar@{=}^{m}_{n} (24,-4.88);(24,4.88)};
(24,5)*\frm<2pt>{*};
(24,-5)*\frm<2pt>{*};
(6,9)*\frm<2pt>{*};
(6,-9)*\frm<2pt>{*};
(18,9)*\frm<2pt>{*};
(18,-9)*\frm<2pt>{*};
{\ar@{=} (6,9);(6,12)};
{\ar@{=} (6,-12);(6,-9)};
{\ar@{=} (18,9);(18,12)};
{\ar@{=} (18,-12);(18,-9)};
\endxy
}

\newcommand{\DoubleLineSpaceA}{
\xy
{\ar@{=} (-6,9);(-0.12,5.12)};
{\ar@{=}_{j_5} (5.86,8.86);(0.12,5.12)};
{\ar@{=} (-6,-9);(-0.12,-5.12)};
{\ar@{=} (5.86,-8.86);(0.12,-5.12)};
{\ar@{=}^{j_1} (0,-4.88);(0,4.88)};
{\ar@{=} (6.14,8.86);(11.86,5.14)};
{\ar@{=}_{j_6} (17.86,8.86);(12.14,5.14)};
{\ar@{=} (6.14,-8.86);(11.86,-5.14)};
{\ar@{=} (17.86,-8.86);(12.12,-5.14)};
{\ar@{=}^{j_2} (12,-4.86);(12,4.86)};
{\ar@{=}^{j_7} (18.2,8.8);(23.8,5.2)};
{\ar@{=} (30,9);(24.12,5.12)};
{\ar@{=} (18.13,-8.87);(23.88,-5.12)};
{\ar@{=} (30,-9);(24.12,-5.12)};
{\ar@{=}^{j_3}_{j_4} (24,-4.88);(24,4.88)};
{\ar@{=} (6,9.13);(6,12)};
{\ar@{=} (6,-11.87);(6,-9.14)};
{\ar@{=} (18,9.13);(18,12)};
{\ar@{=} (18,-12);(18,-9.14)};
\endxy
}

\newcommand{\TriangleBp}{
\xy
0;/r1cm/:;
(0.3,0.15)*{p};
{\aline 0;a(0)};
{\aline 0;a(60)};
{\aline 0;a(120)};
{\aline 0;a(180)};
{\aline 0;a(240)};
{\aline 0;a(300)};
{\aline_>{1} a(0);a(60)};
{\aline_>{2} a(60);a(120)};
{\aline_>{3} a(120);a(180)};
{\aline_>{4} a(180);a(240)};
{\aline_>{5} a(240);a(300)};
{\aline_>{6} a(300);a(0)};
\endxy
}

\newcommand{\TrivalentGraph}{
\xy
0;/r45mm/:;
(0.15,0.84)*{}="p1";
(0.21,0.76)*{}="p2";
(0.58,0.83)*{}="p3";
(0.64,0.90)*{}="p4";
(0.41,0.58)*{}="p5";
(0.73,0.72)*{}="p6";
(0.40,0.37)*{}="p7";
(0.27,0.30)*{}="p8";
(0.42,0.12)*{}="p9";
(0.56,0.20)*{}="p10";
(0.15,0.60)*{}="p11";
(0.81,0.75)*{}="p12";
(0.76,0.42)*{}="p13";
(0.07,0.56)*{}="p14";
(0.1,0.19)*{}="p15";
(0.36,0.08)*{}="p16";
(0.65,0.12)*{}="p17";
(0.84,0.40)*{}="p18";
(0.55,0.55)*{}="p19";
(0.68,0.62)*{}="p20";
(0.42,0.75)*{}="p21";
(0.59,0.30)*{}="p22";
(0.22,0.45)*{}="p23";
(0.20,0.20)*{}="p24";
"p16";"p9" **\dir{-} ?(0.6)*\dir{>}+(0.03,-0.01) *{\scriptstyle j_{1}};
"p17";"p10" **\dir{-} ?(0.6)*\dir{>}+(0.05,0) *{\scriptstyle j_{2}};
"p15";"p24" **\dir{-} ?(0.6)*\dir{>}+(-0.04,0.04) *{\scriptstyle j_{3}};
"p14";"p11" **\dir{-} ?(0.6)*\dir{>}+(-0.03,0.03) *{\scriptstyle j_{4}};
"p2";"p1" **\dir{-} ?(0.6)*\dir{>}+(0.05,0) *{\scriptstyle j_{5}};
"p3";"p4" **\dir{-} ?(0.6)*\dir{>}+(0.04,0) *{\scriptstyle j_{6}};
"p6";"p12" **\dir{-} ?(0.6)*\dir{>}+(0.02,-0.03) *{\scriptstyle j_{7}};
"p18";"p13" **\dir{-} ?(0.6)*\dir{>}+(0.02,0.04) *{\scriptstyle j_{8}};
"p9";"p10" **\dir{-} ?(0.6)*\dir{>}+(-0.02,0.03) *{\scriptstyle j_{9}};
"p22";"p13" **\dir{-} ?(0.6)*\dir{>}+(0.01,-0.04) *{\scriptstyle j_{10}};
"p13";"p20" **\dir{-} ?(0.6)*\dir{>}+(0.04,0) *{\scriptstyle j_{11}};
"p6";"p3" **\dir{-} ?(0.6)*\dir{>}+(0.06,0) *{\scriptstyle j_{12}};
"p9";"p24" **\dir{-} ?(0.6)*\dir{>}+(0,-0.04) *{\scriptstyle j_{13}};
"p8";"p23" **\dir{-} ?(0.6)*\dir{>}+(-0.04,0) *{\scriptstyle j_{14}};
"p11";"p2" **\dir{-} ?(0.6)*\dir{>}+(-0.04,0) *{\scriptstyle j_{15}};
"p21";"p3" **\dir{-} ?(0.6)*\dir{>}+(-0.02,0.04) *{\scriptstyle j_{16}};
"p22";"p7" **\dir{-} ?(0.6)*\dir{>}+(0.03,0.03) *{\scriptstyle j_{17}};
"p8";"p7" **\dir{-} ?(0.6)*\dir{>}+(0.01,-0.03) *{\scriptstyle j_{18}};
"p7";"p19" **\dir{-} ?(0.6)*\dir{>}+(0.04,-0.01) *{\scriptstyle j_{19}};
"p5";"p19" **\dir{-} ?(0.6)*\dir{>}+(0.01,0.04) *{\scriptstyle j_{20}};
"p23";"p5" **\dir{-} ?(0.6)*\dir{>}+(0.03,-0.03) *{\scriptstyle j_{21}};
"p5";"p21" **\dir{-} ?(0.6)*\dir{>}+(-0.04,0) *{\scriptstyle j_{22}};
"p19";"p20" **\dir{-} ?(0.6)*\dir{>}+(0.01,-0.04) *{\scriptstyle j_{23}};
"p20";"p6" **\dir{-} ?(0.6)*\dir{>}+(0.03,-0.02) *{\scriptstyle j_{24}};
"p21";"p2" **\dir{-} ?(0.6)*\dir{>}+(0.02,0.04) *{\scriptstyle j_{25}};
"p10";"p22" **\dir{-} ?(0.6)*\dir{>}+(0.04,-0.01) *{\scriptstyle j_{26}};
"p23";"p11" **\dir{-} ?(0.6)*\dir{>}+(0.04,0.02) *{\scriptstyle j_{27}};
"p24";"p8" **\dir{-} ?(0.6)*\dir{>}+(0.03,-0.02) *{\scriptstyle j_{28}};
  \endxy
  }

\newcommand{\TrivalentGraphII}{
\xy
0;/r45mm/:;
(0.15,0.84)*{}="p1";
(0.21,0.76)*{}="p2";
(0.58,0.83)*{}="p3";
(0.64,0.90)*{}="p4";
(0.41,0.58)*{}="p5";
(0.73,0.72)*{}="p6";
(0.40,0.37)*{}="p7";
(0.27,0.30)*{}="p8";
(0.42,0.12)*{}="p9";
(0.56,0.20)*{}="p10";
(0.15,0.60)*{}="p11";
(0.81,0.75)*{}="p12";
(0.76,0.42)*{}="p13";
(0.07,0.56)*{}="p14";
(0.1,0.19)*{}="p15";
(0.36,0.08)*{}="p16";
(0.65,0.12)*{}="p17";
(0.84,0.40)*{}="p18";
(0.55,0.55)*{}="p19";
(0.68,0.62)*{}="p20";
(0.42,0.75)*{}="p21";
(0.59,0.30)*{}="p22";
(0.22,0.45)*{}="p23";
(0.20,0.20)*{}="p24";
"p16";"p9" **\dir{-} ?(0.6)*\dir{>}+(0.03,-0.01) *{\scriptstyle j_{1}};
"p17";"p10" **\dir{-} ?(0.6)*\dir{>}+(0.05,0) *{\scriptstyle j_{2}};
"p15";"p24" **\dir{-} ?(0.6)*\dir{<}+(-0.04,0.04) *{\scriptstyle j^*_{3}};
"p14";"p11" **\dir{-} ?(0.6)*\dir{<}+(-0.03,0.03) *{\scriptstyle j^*_{4}};
"p2";"p1" **\dir{-} ?(0.6)*\dir{>}+(0.05,0) *{\scriptstyle j_{5}};
"p3";"p4" **\dir{-} ?(0.6)*\dir{>}+(0.04,0) *{\scriptstyle j_{6}};
"p6";"p12" **\dir{-} ?(0.6)*\dir{>}+(0.02,-0.03) *{\scriptstyle j_{7}};
"p18";"p13" **\dir{-} ?(0.6)*\dir{<}+(0.02,0.04) *{\scriptstyle j^*_{8}};
"p9";"p10" **\dir{-} ?(0.6)*\dir{<}+(-0.02,0.03) *{\scriptstyle j^*_{9}};
"p22";"p13" **\dir{-} ?(0.6)*\dir{<}+(0.01,-0.04) *{\scriptstyle j^*_{10}};
"p13";"p20" **\dir{-} ?(0.6)*\dir{>}+(0.04,0) *{\scriptstyle j_{11}};
"p6";"p3" **\dir{-} ?(0.6)*\dir{>}+(0.06,0) *{\scriptstyle j_{12}};
"p9";"p24" **\dir{-} ?(0.6)*\dir{<}+(0,-0.04) *{\scriptstyle j^*_{13}};
"p8";"p23" **\dir{-} ?(0.6)*\dir{<}+(-0.04,0) *{\scriptstyle j^*_{14}};
"p11";"p2" **\dir{-} ?(0.6)*\dir{>}+(-0.04,0) *{\scriptstyle j_{15}};
"p21";"p3" **\dir{-} ?(0.6)*\dir{>}+(-0.02,0.04) *{\scriptstyle j_{16}};
"p22";"p7" **\dir{-} ?(0.6)*\dir{>}+(0.03,0.03) *{\scriptstyle j_{17}};
"p8";"p7" **\dir{-} ?(0.6)*\dir{<}+(0.01,-0.03) *{\scriptstyle j^*_{18}};
"p7";"p19" **\dir{-} ?(0.6)*\dir{<}+(0.04,-0.01) *{\scriptstyle j^*_{19}};
"p5";"p19" **\dir{-} ?(0.6)*\dir{<}+(0.01,0.04) *{\scriptstyle j^*_{20}};
"p23";"p5" **\dir{-} ?(0.6)*\dir{>}+(0.03,-0.03) *{\scriptstyle j_{21}};
"p5";"p21" **\dir{-} ?(0.6)*\dir{>}+(-0.04,0) *{\scriptstyle j_{22}};
"p19";"p20" **\dir{-} ?(0.6)*\dir{<}+(0.01,-0.04) *{\scriptstyle j^*_{23}};
"p20";"p6" **\dir{-} ?(0.6)*\dir{<}+(0.03,-0.02) *{\scriptstyle j^*_{24}};
"p21";"p2" **\dir{-} ?(0.6)*\dir{>}+(0.02,0.04) *{\scriptstyle j_{25}};
"p10";"p22" **\dir{-} ?(0.6)*\dir{>}+(0.04,-0.01) *{\scriptstyle j_{26}};
"p23";"p11" **\dir{-} ?(0.6)*\dir{>}+(0.04,0.02) *{\scriptstyle j_{27}};
"p24";"p8" **\dir{-} ?(0.6)*\dir{<}+(0.03,-0.02) *{\scriptstyle j^*_{28}};
  \endxy
  }

  \newcommand{\Ypart}[6]{
  \bmm\xy
  0;/r0.15pc/:;
  (-10,10)*{}="p1";
  (10,10)*{}="p2";
  (0,2)*{}="p3";
  (0,-10)*{}="p4";
  "p4";"p3" **\dir{-} ?(0.5)*\dir{#4}+(4,0) *{\scriptstyle #1};
  "p3";"p2" **\dir{-} ?(0.6)*\dir{#6}+(2,-3) *{\scriptstyle #3};
  "p3";"p1" **\dir{-} ?(0.6)*\dir{#5}+(-1,-3) *{\scriptstyle #2};
  \endxy\emm
  }

\newcommand{\TriangleYpart}[6]{
  \def\jjA{#1}
  \def\jjB{#2}
  \def\jjC{#3}
  \def\jjD{#4}
  \def\jjE{#5}
  \def\jjF{#6}
  \TriangleYpartExtended
  }

\newcommand{\TriangleYpartExtended}[6]{
  \bmm\xy
  0;/r0.15pc/:;
  (-10,10)*{}="p1";
  (10,10)*{}="p2";
  (0,-10)*{}="p3";
  (0,1.75)*{p};
  (-4,5)*{}="p4";
  (4,5)*{}="p5";
  (0,-3)*{}="p6";
  "p4";"p1" **\dir{-} ?(0.6)*\dir{#5}+(-3,-1) *{\scriptstyle \jjE};
  "p5";"p2" **\dir{-} ?(0.6)*\dir{#6}+(3.5,-1) *{\scriptstyle \jjF};
  "p3";"p6" **\dir{-} ?(0.6)*\dir{#1}+(3,-1) *{\scriptstyle \jjA};
  "p4";"p5" **\dir{-} ?(0.5)*\dir{#4}+(0,3.5) *{\scriptstyle \jjD};
  "p6";"p5" **\dir{-} ?(0.5)*\dir{#3}+(3.5,0) *{\scriptstyle \jjC};
  "p6";"p4" **\dir{-} ?(0.5)*\dir{#2}+(-3,0) *{\scriptstyle \jjB};
  \endxy\emm
}

\newcommand{\Bipartition}{
\renewcommand\latticebody
{\drop{{\ar@{-} c+<-0.866025pc,-0.5pc>;c}
		{\ar@{-} c+<0.866025pc,-0.5pc>;c}{\ar@{-} c;c+<0pc,1pc>}}}
\xy
(0,0)*{\xy
	0;<0.866025pc,0pc>:<0.433013pc,0.75pc>::;
	\croplattice{-5}5{-5}{5}%
	{-3}{3}{-3}{2}
	\endxy};
(0,-12.7)*{}="dd";
{\ar@{-} (-22,5.9);(-25.6,7.8)};
{\ar@{-} (-25.6,7.8);(-25.6,12.1)};
{\ar@{-} (-22,5.9)+"dd";(-25.6,7.8)+"dd"};
{\ar@{-} (-25.6,7.8)+"dd";(-25.6,12.1)+"dd"};
{\ar@{-} (-22,5.9)+"dd"+"dd";(-25.6,7.8)+"dd"+"dd"};
{\ar@{-} (-25.6,7.8)+"dd"+"dd";(-25.6,12.1)+"dd"+"dd"};
{\ar@{-} (22,5.9);(25.6,7.8)};
{\ar@{-} (25.6,7.8);(25.6,12.1)};
{\ar@{-} (22,5.9)+"dd";(25.6,7.8)+"dd"};
{\ar@{-} (25.6,7.8)+"dd";(25.6,12.1)+"dd"};
{\ar@{-} (22,5.9)+"dd"+"dd";(25.6,7.8)+"dd"+"dd"};
{\ar@{-} (25.6,7.8)+"dd"+"dd";(25.6,12.1)+"dd"+"dd"};
(30,0)*{};(-4,-2.5)*{},{\ellipse(17,13.5){--}},
(0,-3)*{A};
(18,-8)*{B};
\endxy
}

\newcommand{\DoubleLineGroundState}{
\renewcommand\latticebody
{\drop{{\ar@{=} c+<-0.82pc,-0.47pc>;c+<-0.03pc,-0.017pc>}
		{\ar@{=} c+<0.82pc,-0.47pc>;c+<0.03pc,-0.017pc>}{\ar@{=} c+<0pc,0.04pc>;c+<0pc,0.96pc>}}}
\xy
(0,0)*{\xy
	0;<0.866025pc,0pc>:<0.433013pc,0.75pc>::;
	\croplattice{-5}5{-5}{5}%
	{-3}{3}{-3}{2}
	\endxy};
(0,-12.7)*{}="dd";
{\ar@{=} (-22.2,6);(-25.4,7.7)};
{\ar@{=} (-25.6,7.9);(-25.6,12.)};
{\ar@{=} (-22.2,6)+"dd";(-25.4,7.7)+"dd"};
{\ar@{=} (-25.6,7.9)+"dd";(-25.6,12.)+"dd"};
{\ar@{=} (-22.2,6)+"dd"+"dd";(-25.4,7.7)+"dd"+"dd"};
{\ar@{=} (-25.6,7.9)+"dd"+"dd";(-25.6,12.)+"dd"+"dd"};
{\ar@{=} (22.2,6);(25.4,7.7)};
{\ar@{=} (25.6,7.9);(25.6,12.)};
{\ar@{=} (22.2,6)+"dd";(25.4,7.7)+"dd"};
{\ar@{=} (25.6,7.9)+"dd";(25.6,12.)+"dd"};
{\ar@{=} (22.2,6)+"dd"+"dd";(25.4,7.7)+"dd"+"dd"};
{\ar@{=} (25.6,7.9)+"dd"+"dd";(25.6,12.)+"dd"+"dd"};
(30,0)*{};(-4,-2.5)*{},{\ellipse(17,13.5){--}},
(-17,3)*{\scriptstyle \alpha_1};
(-15,10)*{\scriptstyle \alpha_2};
(-7,9)*{\scriptstyle \alpha_3};
(-23,-3)*{\scriptstyle \alpha_l};
(-18,-10)*{\scriptscriptstyle \alpha_{l-1}};
\endxy
}

\newcommand{\DiskWithBoundary}{
\renewcommand\latticebody
{\drop{{\ar@{-} c+<-0.866025pc,-0.5pc>;c}
		{\ar@{-} c+<0.866025pc,-0.5pc>;c}{\ar@{-} c;c+<0pc,1pc>}}}
\xy
(0,0)*{\xy
	0;<0.866025pc,0pc>:<0.433013pc,0.75pc>::;
	\croplattice{-5}5{-5}{5}%
	{-3}{3}{-3}{2}
	\endxy};
(0,-12.7)*{}="dd";
(7.333,0)*{}="rr";
{\ar@{-} (-22,5.9);(-25.6,7.8)};
{\ar@{-} (-25.6,7.8);(-25.6,12.1)};
{\ar@{-} (-22,5.9)+"dd";(-25.6,7.8)+"dd"};
{\ar@{-} (-25.6,7.8)+"dd";(-25.6,12.1)+"dd"};
{\ar@{-} (-22,5.9)+"dd"+"dd";(-25.6,7.8)+"dd"+"dd"};
{\ar@{-} (-25.6,7.8)+"dd"+"dd";(-25.6,12.1)+"dd"+"dd"};
{\ar@{-} (22,5.9);(25.6,7.8)};
{\ar@{-} (25.6,7.8);(25.6,12.1)};
{\ar@{-} (22,5.9)+"dd";(25.6,7.8)+"dd"};
{\ar@{-} (25.6,7.8)+"dd";(25.6,12.1)+"dd"};
{\ar@{-} (22,5.9)+"dd"+"dd";(25.6,7.8)+"dd"+"dd"};
{\ar@{-} (25.6,7.8)+"dd"+"dd";(25.6,12.1)+"dd"+"dd"};
{\ar@{-} (-22,18.4);(-18.333,20.52)};
{\ar@{-} (-14.667,18.4);(-18.333,20.52)};
{\ar@{-} (-22,18.4)+"rr";(-18.333,20.52)+"rr"};
{\ar@{-} (-14.667,18.4)+"rr";(-18.333,20.52)+"rr"};
{\ar@{-} (-22,18.4)+"rr"+"rr";(-18.333,20.52)+"rr"+"rr"};
{\ar@{-} (-14.667,18.4)+"rr"+"rr";(-18.333,20.52)+"rr"+"rr"};
{\ar@{-} (-22,18.4)+"rr"+"rr"+"rr";(-18.333,20.52)+"rr"+"rr"+"rr"};
{\ar@{-} (-14.667,18.4)+"rr"+"rr"+"rr";(-18.333,20.52)+"rr"+"rr"+"rr"};
{\ar@{-} (-22,18.4)+"rr"+"rr"+"rr"+"rr";(-18.333,20.52)+"rr"+"rr"+"rr"+"rr"};
{\ar@{-} (-14.667,18.4)+"rr"+"rr"+"rr"+"rr";(-18.333,20.52)+"rr"+"rr"+"rr"+"rr"};
{\ar@{-} (-22,18.4)+"rr"+"rr"+"rr"+"rr"+"rr";(-18.333,20.52)+"rr"+"rr"+"rr"+"rr"+"rr"};
{\ar@{-} (-14.667,18.4)+"rr"+"rr"+"rr"+"rr"+"rr";(-18.333,20.52)+"rr"+"rr"+"rr"+"rr"+"rr"};
\endxy
}

\newcommand{\DoubleLineDisk}{
\renewcommand\latticebody
{\drop{{\ar@{=} c+<-0.82pc,-0.47pc>;c+<-0.03pc,-0.017pc>}
		{\ar@{=} c+<0.82pc,-0.47pc>;c+<0.03pc,-0.017pc>}{\ar@{=} c+<0pc,0.04pc>;c+<0pc,0.96pc>}}}
\xy
(0,0)*{\xy
	0;<0.866025pc,0pc>:<0.433013pc,0.75pc>::;
	\croplattice{-5}5{-5}{5}%
	{-3}{3}{-3}{2}
	\endxy};
(0,-12.67)*{}="dd";
(0,12.67)*{}="uu";
(7.318,0)*{}="rr";
(-7.318,0)*{}="ll";
(0,-0.4)*{}="eedd";
{\ar@{=} (-22.2,5.96);(-25.4,7.7)};
{\ar@{=} (-25.6,7.9);(-25.6,12.)};
{\ar@{=} (-22.2,5.96)+"dd";(-25.4,7.7)+"dd"};
{\ar@{=} (-25.6,7.9)+"dd";(-25.6,12.)+"dd"};
{\ar@{=} (-22.2,5.96)+"dd"+"dd";(-25.4,7.7)+"dd"+"dd"};
{\ar@{=} (-25.6,7.9)+"dd"+"dd";(-25.6,12.)+"dd"+"dd"};
{\ar@{=} (22.2,5.96);(25.4,7.7)};
{\ar@{=} (25.6,7.9);(25.6,12.)};
{\ar@{=} (22.2,5.96)+"dd";(25.4,7.7)+"dd"};
{\ar@{=} (25.6,7.9)+"dd";(25.6,12.)+"dd"};
{\ar@{=} (22.2,5.96)+"dd"+"dd";(25.4,7.7)+"dd"+"dd"};
{\ar@{=} (25.6,7.9)+"dd"+"dd";(25.6,12.)+"dd"+"dd"};
{\ar@{=} (-21.8,18.5);(-18.533,20.52)};
{\ar@{=} (-14.767,18.5);(-18.133,20.52)};
{\ar@{=} (-21.8,18.5)+"rr";(-18.533,20.52)+"rr"};
{\ar@{=} (-14.767,18.5)+"rr";(-18.133,20.52)+"rr"};
{\ar@{=} (-21.8,18.5)+"rr"+"rr";(-18.533,20.52)+"rr"+"rr"};
{\ar@{=} (-14.767,18.5)+"rr"+"rr";(-18.133,20.52)+"rr"+"rr"};
{\ar@{=} (-21.8,18.5)+"rr"+"rr"+"rr";(-18.533,20.52)+"rr"+"rr"+"rr"};
{\ar@{=} (-14.767,18.5)+"rr"+"rr"+"rr";(-18.133,20.52)+"rr"+"rr"+"rr"};
{\ar@{=} (-21.8,18.5)+"rr"+"rr"+"rr"+"rr";(-18.533,20.52)+"rr"+"rr"+"rr"+"rr"};
{\ar@{=} (-14.767,18.5)+"rr"+"rr"+"rr"+"rr";(-18.133,20.52)+"rr"+"rr"+"rr"+"rr"};
{\ar@{=} (-21.8,18.5)+"rr"+"rr"+"rr"+"rr"+"rr";(-18.533,20.52)+"rr"+"rr"+"rr"+"rr"+"rr"};
{\ar@{=} (-14.767,18.5)+"rr"+"rr"+"rr"+"rr"+"rr";(-18.133,20.52)+"rr"+"rr"+"rr"+"rr"+"rr"};
@={ (-18.35,21.1)="oo",(-14.65,18.8)="pp", "oo"+"rr","pp"+"rr",
	"oo"+"rr"+"rr","pp"+"rr"+"rr",
	"oo"+"rr"+"rr"+"rr","pp"+"rr"+"rr"+"rr",
	"oo"+"rr"+"rr"+"rr"+"rr","pp"+"rr"+"rr"+"rr"+"rr",
	"oo"+"rr"+"rr"+"rr"+"rr"+"rr",(22.3,18.6)="o1",
	(22.3,14.4)="o2",(25.9,12.35)="o3",(25.9,7.6)="o4",
	"o1"+"eedd"+"dd","o2"+"dd","o3"+"dd","o4"+"dd",
	"o1"+"eedd"+"dd"+"dd","o2"+"dd"+"dd","o3"+"dd"+"dd","o4"+"dd"+"dd",
	(21.95,-19.95)="b1",(18.4,-17.8)="b2",
	"b1"+"ll","b2"+"ll",
	"b1"+"ll"+"ll","b2"+"ll"+"ll",
	"b1"+"ll"+"ll"+"ll","b2"+"ll"+"ll"+"ll",
	"b1"+"ll"+"ll"+"ll"+"ll","b2"+"ll"+"ll"+"ll"+"ll",
	"b1"+"ll"+"ll"+"ll"+"ll"+"ll","b2"+"ll"+"ll"+"ll"+"ll"+"ll",
	"b1"+"ll"+"ll"+"ll"+"ll"+"ll"+"ll",(-25.95,-17.8)="a1",
	(-25.95,-13)="a2",(-22.2,-10.8)="a3",(-22.2,-7.1)="a4",
	"a1"+"uu","a2"+"uu","a3"+"uu","a4"+"uu",
	"a1"+"uu"+"uu","a2"+"uu"+"uu","a3"+"uu"+"uu",(-22.3,18.6)};
s0="prev";
@@{;"prev";**@{-}="prev"};
\endxy
} 

\usepackage{amscd,amstext,mathrsfs,CJK,latexsym}
\usetikzlibrary{intersections,decorations.pathmorphing,decorations.fractals}

\newcommand{\aline}{\ar@{-}}
\newcommand{\arl}{\ar@{-}|@{>}}
\newcommand{\arr}{\ar@{-}|@{<}}
\newcommand{\rmd}{\mathrm{d}}
\newcommand{\rmw}{\mathrm{w}}
\newcommand{\rmv}{\mathrm{v}}

\newcommand{\dsI}{\mathbf{1}}
\newcommand{\dsC}{\mathbb{C}}
\newcommand{\mcc}{\mathcal{C}}
\newcommand{\mcb}{\mathcal{B}}
\newcommand{\dsR}{\mathbb{R}}
\newcommand{\cM}{\mathcal{M}}
\newcommand{\scH}{\mathcal{L}}

\newcommand{\ket}[1]{\left|{#1}\right\rangle}
\newcommand{\bra}[1]{\left\langle{#1}\right|}

\newcommand{\bpm}{\begin{pmatrix}}
\newcommand{\epm}{\end{pmatrix}}
\newcommand{\bmm}{\begin{matrix}}
\newcommand{\emm}{\end{matrix}}
\newcommand{\mat}[2][1]{%
  \scalebox{#1}{%
    \renewcommand{\arraystretch}{1}%
    $\begin{pmatrix}#2\end{pmatrix}$%
  }
}

\newcommand{\M}{\mathcal{M}}

\newcommand{\V}{\mathcal{V}}

\newcommand{\Z}{\mathbb{Z}}
\newcommand{\N}{\mathbb{N}}

\def\Cen{\mathcal{Z}}

\newcommand{\two}{\mathbf{2}}
\newcommand{\C}{\mathbb C}

\numberwithin{equation}{section}

\newtheorem{theorem}{Theorem}[section]

\newtheorem{conj}[theorem]{Conjecture}

\newtheorem{prop}[theorem]{Proposition}
\theoremstyle{definition}

\newtheorem{remark}[theorem]{Remark}

\newtheorem{example}[theorem]{Example}
\newtheorem{definition}[theorem]{Definition}

\begin{document}
\title[Symmetry Enriched Topological Phases]
{On Enriching the Levin-Wen model with Symmetry}

\author{Liang Chang$^{1}$, Meng Cheng$^{2}$, Shawn X. Cui$^{3}$, Yuting Hu$^{4}$, Wei Jin$^{5}$, Ramis Movassagh$^{6}$, Pieter Naaijkens$^{7}$, Zhenghan Wang$^{2,3}$, and Amanda Young$^{8}$}
\address{$^1$Department of Mathematics\\Texas A$\&$M University\\College Station, Texas 77843-1224\\USA}
\email{liangchang@math.tamu.edu}
\address{$^2$Microsoft Research, Station Q\\ UC at Santa Barbara, CA 93106\\USA}
\email{mcheng@microsoft.com, zhenghwa@microsoft.com}
\address{$^3$Department of Mathematics\\UC at Santa Barbara, CA 93106\\USA}
\email{xingshan@math.ucsb.edu, zhenghwa@math.ucsb.edu}
\address{$^4$Department of Physics and Astronomy\\University of Utah\\Salt Lake City, UT 84112\\USA}
\email{yuting@physics.utah.edu}
\address{$^5$School of Mathematical Sciences\\Peking University\\Beijing 100871\\China}
\email{weijin@math.ucsb.edu}
\address{$^6$Department of Mathematics\\Northeastern University\\Boston, MA 02115\\USA}
\email{ramis.mov@gmail.com}
\address{$^7$Institut f\"{u}r Theoretische Physik\\Leibniz Universit\"{a}t Hannover\\Germany}
\email{pieter.naaijkens@itp.uni-hannover.de}
\address{$^8$Department of Mathematics\\UC at Davis\\Davis, CA 95616\\USA}
\email{amyoung@math.ucdavis.edu}

\thanks{The authors thank AMS and NSF for sponsoring the 2014 Mathematical
Research Community on mathematics of quantum phases of matter and quantum
information in Snowbird, Utah, where this group project began.  The third and
eighth named authors are partially supported by NSF DMS 1108736.  The fifth
named author is supported by the China Scholarship Council. The seventh named
author is supported through NWO Rubicon grant 680-50-1118 and the EU project
QFTCMPS. The ninth named author is supported by NSF DMS 1009502.}


\begin{abstract}
Symmetry protected and symmetry enriched topological phases of matter are of
great interest in condensed matter physics due to new materials such as
topological insulators. The Levin-Wen model for spin/boson systems is an
important rigorously solvable model for studying $2D$ topological phases.  The
input data for the Levin-Wen model is a unitary fusion category, but the same
model also works for unitary multi-fusion categories.  In this paper, we provide
the details for this extension of the Levin-Wen model, and show that the
extended Levin-Wen model is a natural playground for the theoretical study of
symmetry protected and symmetry enriched topological phases of matter.

\end{abstract}

\maketitle

\section{Introduction}

Symmetry protected and symmetry enriched topological phases of matter are of
great interest in condensed matter physics due to new materials such as
topological insulators (see \cite{CGLW,BBCW} and references therein).  The Levin-Wen
(LW) model for spin/boson systems is an important rigorously solvable model for
studying $2D$ topological phases \cite{LW}.  The required input data for the LW
model is a unitary fusion category (UFC), but the same model works for unitary
multi-fusion categories.  In this paper, we provide several results for this
extension of the LW model, and show that the extended LW model is a natural
playground for the theoretical study of symmetry protected and symmetry enriched
topological phases of matter in two spatial dimensions.

The LW model is a Hamiltonian formulation of Turaev-Viro $(2+1)$-TQFTs.  Three
mathematical theorems underlie this beautiful model:  (1) given a UFC $\mcc$,
we can construct a Turaev-Viro unitary $(2+1)$-TQFT \cite{BW}, (2) the
Drinfeld center $\Cen(\mcc)$ or quantum double $D(\mcc)$ of a UFC $\mcc$ is
always modular \cite{Mu}, and (3) the Turaev-Viro $(2+1)$-TQFT based on
$\mcc$ is equivalent to the Reshetikhin-Turaev $(2+1)$-TQFT based on the center
$\Cen(\mcc)$ \cite{BK,TV}.  The algebraic model of anyons in the LW model with
input $\mcc$ is encoded by the modular category $\Cen(\mcc)$.

We conjecture that all three theorems above have appropriate extensions to
unitary multi-fusion categories.  Indeed the Drinfeld center $\Cen(\mcc)$ of an
indecomposable multi-fusion category $\mcc$ is modular, and a direct sum of
modular categories if $\mcc$ is decomposable.  Thus, we expect the Hilbert space
$V(S^2)$ of the $2$-sphere $S^2$ associated to a decomposable multi-fusion
category $\mcc$ has dimension $>1$.

There are several generalizations of the LW model, including to $3D$ and fermion
systems \cite{WW,GWW}.  The first appearance of a LW model using a unitary
multi-fusion category as input is given in Example $H$ of Section III
in \cite{LWYW}. While the extension of the LW model to unitary multi-fusion
categories as input is straightforward, the application of this extension to
symmetry protected and symmetry enriched topological phases of matter is new.

In $2D$, the anyon model of a topological phase of quantum matter is
algebraically modeled by a unitary modular category $\mcb$.  An exciting new
direction is the interplay between symmetry and topological order \cite{BBCW}.
But a microscopic physical theory based on local Hamiltonians is still lacking.
For topological phases such that $\mcb$ is a quantum double $\mcb=D(\mcc)$, the
LW model could provide such a microscopic theory. Specifically, given an input
$\mcc$ for the LW model, if the symmetry $G$ could be realized as unitary on-site
symmetries of the LW Hamiltonians, then the topological symmetry on $D(\mcc)$
should emerge from the $G$ symmetry of the Hamiltonians.  But even for the
electric-magnetic duality $e\leftrightarrow m$ of the toric code, a Hamiltonian
realization is not in the literature\footnote{Meng Cheng found an on-site
realization of the electric-magnetic duality in the toric code, but the details
have not been published.}.  Current realizations of the $e\leftrightarrow m$
duality need the dual lattice and lattice translation.

In the case of a multi-fusion category, group symmetries sometimes appear in a
natural way. For such a category it is natural to consider labels consisting of
two indices.  We may then endow the half-labels with a group structure $G$.
Then the solutions of pentagons are closely related to $G$-equivariant
$3$-cocycles, and extended LW Hamiltonians sometimes naturally come with a
$G$-symmetry, as we will see below.  This leads to an application of the LW model to 
symmetry protected and symmetry enriched topological phases.

The contents of the paper are as follows:  In Sec.~2, we provide some background
material on multi-fusion categories.  In Sec.~3, we give the detail of the
extension of the LW model to multi-fusion category inputs and prove that the
extended LW models with input $\cM_n$ all realize the trivial $(2+1)$-TQFT.  In
Sec.~4, we introduce group structures onto the half-label set of a multi-fusion
category and use such group structures to enrich the LW model with symmetries.
Finally, we de-equivariantize our $G$-symmetric LW models with a non-local
transformation that leads to traditional LW models coupled with a local group
action.

\section{Multi-fusion categories and their doubles}

All multi-fusion and modular categories in this paper are unitary over the
complex numbers $\mathbb{C}$.

\subsection{Multi-fusion category}\label{multi-fusion}

The tensor unit is required to be a simple object in a fusion category.  If we
allow the tensor unit to be not necessarily simple, we obtain multi-fusion
categories.  Therefore, a multi-fusion category is a finite semi-simple rigid
monoidal $\C$-linear category.  They arise naturally in mathematics and physics.
  For example, given a finite depth type $\Pi_1$ sub-factor $N\subset M$ in the
study of von Neummann algebras, the $N-N$, $N-M$, $M-N$, and $M-M$ bi-modules
form a Morita context, and can be regarded as a multi-fusion category.
Much of the fusion category theory naturally generalizes to the multi-fusion
case.

Given a multi-fusion category $\mcc$ with a tensor unit $\bf{1}$, the tensor
unit $\bf{1}$ decomposes into the sum of simple objects ${\bf{1}}\cong
\oplus_{i=1}^n {\bf{1}_i}$ for some $n$.
For a simple object $X$ of $\mcc$, there exists a unique pair $1\leq i,j\leq n$
such that ${\bf{1}_{i}} \otimes X\cong X\cong X\otimes {\bf{1}_j}$.  We will say
that $X$ is in the $(i,j)$-th component of $\mcc$. Let $\mcc_{ij}$ be the
abelian\footnote{Here we mean ``abelian'' as in the sense it is used in category
theory and homology theory, not as in abelian anyons.} sub-category of $\mcc$
generated by direct sums of all simple objects in the $(i,j)$-th component.  We
will call $\mcc_{ij}$ the $(i,j)$-th component of $\mcc$.  The diagonal
components $\mcc_{ii}$ are fusion categories and the off-diagonal components
$\mcc_{ij}, i\neq j,$ are $\mcc_{ii}$-$\mcc_{jj}$-bimodules.
We will call such a multi-fusion category an $n\times n$ multi-fusion category.
A $1\times 1$ multi-fusion category is just a fusion category.
A multi-fusion category is indecomposable if it is not the direct sum of two
non-zero multi-fusion categories.

\begin{definition}

An $n\times n$ $\two$-matrix is an $n\times n$ multi-fusion category for which
each component $\mcc_{i,j}$ is equivalent to $\V ec$, and the fusion rule is
$E_{ij}\otimes E_{kl}=\delta_{jk}E_{il}$,  where $\{E_{ij}\}_{1\leq i,j\leq n}$
is a complete set of isomorphism classes of all simple objects.  We will call
$\{ i \}_{1\leq i \leq n}$ the half-label set.

\end{definition}

\begin{example}\label{example:matrices}
The  $n\times n$ $\two$-matrix $\mathcal{M}_n$.

The multi-fusion category $\cM_n$ is the semi-simple category with simple
objects $\{E_{ij}\}, 1 \leq i,j \leq n,$ and fusion rule $E_{ij} \otimes E_{kl}
= \delta_{jk}E_{il}.$  The tensor product is strictly associative as matrix
multiplication, and the tensor unit is $\dsI=\oplus_{i=1}^n E_{ii}$.  $\cM_n$
can be regarded as a categorification of the matrix algebra $M_n$ by replacing
$\C$ with $\V ec$.

A general object in $\cM_n$ is of the form $X = \bigoplus\limits_{i,j=1}^{n}
x_{ij}E_{ij}, x_{ij} \in \mathbb{N}$.  The multiplicities $x_{ij}$ will be
assembled into an $n \times n$ matrix, denoted also as $X$. So an object $X$ is
given by an $n \times n$ matrix $X=(x_{ij})_{1\leq i,j\leq n}$ with non-negative
integral entries, and $E_{ij}$ is represented by the matrix as the notation
indicates: all entries are zero except the $(i,j)$-entry, which is $1$. Then the
tensor product of two objects $X,Y$ is just the matrix multiplication $XY$. For
$X = (x_{ij}), Y = (y_{ij}),$ a morphism from $X$ to $Y$ is of the form $f =
(f_{ij}), $ where $f_{ij}: x_{ij}E_{ij} \longrightarrow y_{ij}E_{ij}$ can be
represented by a linear map from $\mathbb{C}^{x_{ij}} \longrightarrow
\mathbb{C}^{y_{ij}}$, or simply a $y_{ij} \times x_{ij}$ matrix. Hence, a
morphism in $\mathcal{M}_n$ is simply a matrix of matrices.  Then compositions
of morphisms are given by entry-wise matrix multiplication.

\end{example}

\begin{example}
Morita contexts as multi-fusion categories.

Suppose $\mcc$ is a fusion category and $\M$ an indecomposable module category
over $\mcc$.  Let $\mcc^{*}_{\M}={F}un_{\mcc}(\M,\M)$ be the dual of $\mcc$ with
respect to $\M$.  Then
$\begin{pmatrix}
\mcc&{\M}^{*}\\
\M&\mcc^{*}_{\M}
\end{pmatrix}$
is a $2\times 2$ multi-fusion category.

\end{example}

\subsection{Quantum Doubles}

Suppose $\mcc$ is a multi-fusion category, then its quantum double $D(\mcc)$ in
physics or Drinfeld center $\Cen(\mcc)$ in mathematics is also a multi-fusion
category.  Note that $D(\mcc_1\oplus\mcc_2)\cong D(\mcc_1)\oplus D(\mcc_2)$ for
two multi-fusion categories $\mcc_i, i=1,2$.  Therefore, we will mainly focus on
indecomposable multi-fusion categories.

\begin{theorem}\label{double}

Let $\mcc=(\mcc_{ij})_{1\leq i,j\leq n}$ be an $n\times n$ indecomposable
multi-fusion category.  Then the quantum double $D(\mcc)$ of $\mcc$ is
equivalent to $D(\mcc_{ii})$ for any $1\leq i \leq n$. It follows that all
$\mcc_{ii}$ are categorically Morita equivalent to each other.
\end{theorem}

\begin{proof}

If $\M$ is an indecomposable module category over an indecomposable multi-fusion
category $\mcc$, then $D(\mcc)=D(\mcc^{*}_{\M})$, where $\mcc^{*}_{\M}$ is the
dual of $\mcc$ with respect to $\M$ (Corollary 3.35 \cite{EO}).  For a fixed
$i$, let $\M_{i}=\oplus_{k=1}^n \mcc_{ik}$.  Then $\M_i$ is an indecomposable
$\mcc$-module category.  The dual category of $\mcc$ with respect to $\M_i$ is
$\mcc^{*}_{\M_i}\cong \mcc_{ii}^{\textrm{op}}$, where $\mcc_{ii}^{\textrm{op}}$
is the opposite category of $\mcc$.  The theorem now follows from
$D(\mcc)\cong D(\mcc^{*}_{\M_i})\cong D(\mcc_{ii}^{\textrm{op}})\cong D(\mcc)$.
\end{proof}

\subsection{Doubles of $n\times n$ $\two$-matrices $\cM_n$}

It follows from Thm.~\ref{double} that $D(\cM_n) \cong \V ec$.  To keep our presentation elementary, we provide an
explicit proof that $D(\cM_n)$ is $\V ec$ in this subsection.

Suppose $X = (x_{ij}) = \bigoplus x_{ij}E_{ij}$ is an object of $\cM_n$, and $(X, c_{X},{}_{-})$ an object of $D(\cM_n)$.  Then for any $E_{ij},$ $c_{X, E_{ij}}: X \otimes E_{ij}  \longrightarrow E_{ij} \otimes X$ is an isomorphism.  Since $X \otimes E_{ij} = \bigoplus\limits_{k=1}^{n} x_{ki}E_{kj},$ and $E_{ij} \otimes X = \bigoplus\limits_{k=1}^{n} x_{jk}E_{ik}$, we have $x_{ki} = 0, k\neq i,$ and $x_{ii} = x_{jj}$ for any pair $i,j$. Write $x_{ii} = m$, then $X \otimes E_{ij} = m E_{ij} = E_{ij} \otimes X$, and $c_{X, E_{ij}}$ is an $n \times n$ matrix whose $(i,j)$-entry is an isomorphism $m E_{ij} \longrightarrow m E_{ij}$, i.e. a matrix in $GL(m, \C)$, and whose other entries are all $0$. Thus an object of $D(\cM_n)$ is determined by the set $\{(m, c_{ij})\}, 1 \leq i,j \leq n,$ where $m$ is a positive integer, and $c_{ij} \in GL(m, \C)$. Explicitly, $ X = m I_n$ , and the half braiding between $X$ and $E_{ij}$ is $c_{ij}: m E_{ij} \longrightarrow m E_{ij}$.

\begin{figure}
\begin{tikzpicture}[scale = 0.5]
\begin{scope}
   \draw (4,0) node[left]{$E_{ij}$}-- (1,4);
   \draw (6,0) node[left]{$E_{kl}$}-- (3,4);
   \draw (1,0) [color = white, line width = 2mm] -- (6,4);
   \draw (1,0) node[left]{$X$} -- (6,4);
   \draw (7,2) node{$=$};
   \draw (10,0) node[right]{$E_{ij}$} -- (8,2)-- (8,4);
   \draw (13,0) node[right]{$E_{kl}$} -- (13,2)-- (10,4);
   \draw (8,0)[color = white, line width = 2mm] -- (10,2)[color = white, line width = 2mm] -- (13,4);
   \draw (8,0) node[right]{$X$}  -- (10,2) -- (13,4);
\end{scope}
\end{tikzpicture}
\caption{$\textrm{Hexagon Equations}$} \label{hexagon}
\end{figure}
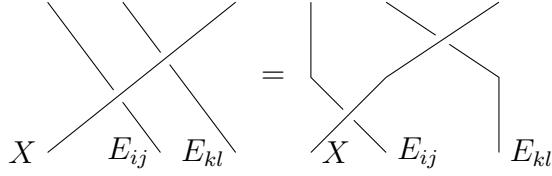

To find the constraints from the hexagon equations as illustrated by Fig.~\ref{hexagon},
we see that the left-hand side of the equation in Fig.~\ref{hexagon} is given by
$\delta_{jk}c_{il}: m E_{il} \longrightarrow m E_{il},$ and the right-hand side
is given by $\delta_{jk} c_{jl}c_{ij}$. Thus we obtain

\begin{equation} \label{braiding equ}
c_{ij} = c_{kj}c_{ik}, \forall 1 \leq i,j,k \leq n.
\end{equation}

Since every $c_{ij}$ is invertible, it follows that $c_{ii} = I_{m}$, and
$c_{ij} = c_{ji}^{-1}$. Hence the $c_{ij}$'s are completely determined by
$c_{i1}, 2\leq i\leq n$ through the formula $c_{ij} = c_{j1}^{-1}c_{i1}$. The
matrices $c_{21}, \cdots, c_{n1} \in GL(m, \C)$ can be chosen arbitrarily, and $c_{11}=I_m$. Thus, an object of $D(\cM_n)$ is determined
by a positive integer $m$ and $(n-1)$ matrices $c_{21}, \cdots, c_{n1} \in GL(m,
\C)$.

To understand the morphisms in the doubles, we consider two objects $(X,
c_{ij}), (X', c_{ij}'),$ where $X = m I_n, X' = m' I_n$.  Then a morphism
$\varphi:(X, c_{ij})\rightarrow (X', c_{ij}')$ is given by
$(\delta_{ij}\varphi_{ii}),$ where $\varphi_{ii}: mE_{ii} \longrightarrow
m'E_{ii}$ is a linear map. This morphism should commute with the half braiding,
shown in Fig.~\ref{morphism}.

\begin{figure}[!h]
\begin{tikzpicture}[scale = 0.5]
\begin{scope}
   \draw (3,0) node[left]{$E_{ij}$}-- (0,3) -- (0,5) node[left]{$E_{ij}$};
   \draw (0,0) [color = white, line width = 2mm] -- (3,3);
   \draw (0,0) node[left]{$X$} -- (3,3) node[right]{$X$} -- (3,3.5);
   \draw (2.5,3.5) rectangle (3.5,4.5);
   \draw (3,4.5) -- (3,5) node[right]{$X'$};
   \draw (3,4) node{$\varphi$};
   \draw (5,3) node{$=$};
   \draw (10,0) node[left]{$E_{ij}$} -- (10,2) -- (7,5) node[left]{$E_{ij}$};
   \draw (7,0) node[left]{$X$} -- (7,0.5);
   \draw (6.5,0.5) rectangle (7.5,1.5);
   \draw (7,1) node{$\varphi$};
   \draw (7,1.5) -- (7,2) node[left]{$X'$} [color = white, line width = 2mm] -- (10,5);
   \draw (7,1.5) -- (7,2) node[left]{$X'$} -- (10,5) node[left]{$X'$};
\end{scope}
\end{tikzpicture}
\caption{$\textrm{Morphisms in } D(\cM_n)$} \label{morphism}
\end{figure}

Fig.~\ref{morphism} leads to the following equations for the morphism $\varphi$
to satisfy:
$$\varphi_{jj} c_{ij} = c_{ij}' \varphi_{ii}.$$

Now assume $m = m', $ and $\varphi_{ii}$ is an isomorphism.  The equations above
can be rewritten as $c_{ij}' = \varphi_{jj}c_{ij}\varphi_{ii}^{-1}$. By
Eq.~(\ref{braiding equ}), it suffices to satisfy $c_{i1}' =
\varphi_{11}c_{i1}\varphi_{ii}^{-1}$ for $i=2,\cdots n .$  Using the freedom for
choosing $\varphi_{ii}$, we choose them so that $c_{i1}' = I_{m}$ for all $i$,
and thus $c_{ij}' = I_m, \forall 1 \leq i,j \leq n$. Therefore, two objects of
$D(\cM_n)$ are isomorphic if and only if their diagonal entries $m$ and $m'$ are
the same, i.e. an isomorphism class is uniquely determined by a positive integer
$m$.  For each $m$,  we choose a representative $(X,c_{ij}) = (m I_n, I_m),$
which is denoted as $(m)$.

Note that $(m) \oplus (m') = (m+m')$. Hence, $D(\cM_n)$ is generated by the
single object $(1)=(I_n,1)$. Note that $Hom((1),(1)) = \C,$ so $(1)$ is the only
simple object in the category. Thus, $D(\cM_n) = \V ec$ as expected.

\section{Levin-Wen model for Multi-fusion Categories}

Fix an integer $d\geq 2$, and a cellulation $\gamma$ of an oriented closed
surface $Y$.  We often also refer to $\gamma$ as a graph in $Y$ by thinking
about the $1$-skeleton of $\gamma$. Let $V(\gamma), E(\gamma)$, and $F(\gamma)$
be the set of vertices (sites), edges (bonds), and faces (plaquettes) of
$\gamma$, respectively.
Then $L_{\gamma}(Y)$ will be the local Hilbert space $\otimes_{e\in E(\gamma)}\C^d$,
i.e. we attach a {\it qudit} $\C^d$ to each edge.  The orthonormal basis of
$L_{\gamma}(Y)$ consists of all colors of the edges by a basis of $\C^d$.
In this section, $d$ will be the rank of the input UFC $\mcc$, i.e., the number
of labels.

\begin{definition}

A Hamiltonian $H$ is a \textit{commuting local projector} (CLP) Hamiltonian if
$H=\sum_{\alpha}P_{\alpha}$, where $P_{\alpha}$ is a collection of pair-wise
commuting local orthogonal projectors.

\end{definition}

In general, we are not really interested in a single CLP Hamiltonian, rather a
prescription for writing down a family of CLP Hamiltonians on all local Hilbert
spaces $L_{\gamma}(Y)$ associated to  cellulations $\gamma$ of $Y$.  Such a
prescription will be called a {\it Hamiltonian schema}.  Since we are interested
in thermodynamical physics, we need to study limits when the size of cellulations
measured by the mesh goes to $0$.  We can use Pachner's theorem to organize
all triangulations of a surface into a directed set.  Then local Hilbert
spaces and their ground state manifolds form inverse systems of finite dimensional
Hilbert spaces.

The numerical data to specify the local Hilbert space and Hamiltonian of a LW
model is a description of a UFC in terms of $6j$-symbols.  In order to
implement unitarity and symmetries, we demand some symmetries of the $6j$ symbols.  There
are subtleties when the input UFC has multiplicities in the fusion rules, as defined below, and
non-trivial Frobenius-Schur indicators.  In the following, we will assume that
all UFCs are multiplicity free and their modified $6j$-symbols, called
tetrahedral symbols, have the full tetrahedral symmetry, as defined below.  Not
all UFCs have tetrahedral symbols that have the full tetrahedra symmetry
\cite{Ho}.

\subsection{Levin-Wen Hamiltonian schema for unitary fusion categories}

A \textit{label set} $L$ is a finite set with a distinguished element $0$ and
with an involution $^*:L\rightarrow L$ such that $0^*=0$. Elements of $L$ are
called labels, $0$ is called the trivial label, and $j^*\in L$ is called the
dual of $j\in L$.

A \textit{fusion rule} on $L$ is $N:L \times L \times L\rightarrow \N$ such that for $a,b,c,d\in L$,
\begin{align}
&N_{0a}^b=N_{a0}^b=\delta_{ab},\label{eq:N:a0b}\\
&N_{ab}^0=\delta_{ab^*},\label{eq:N:ab0}\\
&\sum_{x\in L}N_{ab}^xN_{xc}^d=\sum_{x\in L}N_{ax}^dN_{cd}^x.\label{eq:NN=NN}
\end{align}
A fusion rule is \textit{multiplicity-free} if $N_{ab}^c\in \{0,1\}$ for all
$a,b,c\in L$. Set $\delta_{abc}:=N_{ab}^{c^*}$, then $\delta_{abc}=\delta_{bca}$ and
$\delta_{abc}=\delta_{c^*b^*a^*}$. A triple $(a,b,c)$ is admissible if
$\delta_{abc}=1$.

Given a fusion rule on $L$, a \textit{loop weight} is a map $\rmw:L\rightarrow
\dsR \backslash \{0\}$ such that
$\rmw_{a^*}=\rmw_{a}$ and
\begin{equation}
\sum_{c\in L}\rmw_c\delta_{abc^*}=\rmw_a\rmw_b.
\end{equation}
In particular, $\rmw_0=1$. For unitary modular categories, the quantum
dimensions---quantum traces of the identity morphisms---satisfy $\rmd_j\geq 1$ for all $j\in L$. Quantum dimensions might differ from
loop weights $\{w_i\}$.  We let $\alpha_i=\frac{\rmd_i}{\rmw_i}=\pm 1$ for each label, and require:
\begin{equation}
\alpha_i\alpha_j\alpha_k=1, \quad \text{if }\delta_{ijk}=1.
\end{equation}

A \textit{symmetrized tetrahedral symbol} is a map $T:L^6\rightarrow \dsC$ satisfying the following conditions:
\begin{align}
&\text{tetrahedral symmetry:} &T^{ijm}_{kln}=T^{mij}_{nk^{*}l^{*}}
=T^{klm^{*}}_{ijn^{*}}=\alpha_m\alpha_n\,\overline{T^{j^*i^*m^*}_{l^*k^*n}},\label{eq:6j:TetSymmetry}\\
&\text{pentagon identity:} &\sum_{n}{\rmw_{n}}T^{mlq}_{kp^{*}n}T^{jip}_{mns^{*}}T^{js^{*}n}_{lkr^{*}}
=T^{jip}_{q^{*}kr^{*}}T^{riq^{*}}_{mls^{*}},\label{eq:6j:PentagonIdentity}\\
&\text{orthogonality condition:} &\sum_{n}{\rmw_{n}}T^{mlq}_{kp^{*}n}T^{l^{*}m^{*}i^{*}}_{pk^{*}n}
=\frac{\delta_{iq}}{\rmw_{i}}\delta_{mlq}\delta_{k^{*}ip},\label{eq:6j:OrthogonalityCondition}
\end{align}

For convenience, we consider LW models defined on trivalent graphs in a closed oriented surface. Initially, we choose an arrow of each edge to assign a label, but the Hilbert space does not depend on these arrows, by using the following identification: for any state $|\psi\rangle \in L_{\gamma}(Y)$, if we reverse the direction of an edge $e$ and replace its label $j_e$ by its dual $j_e^*$, then the resulting state is identified with the initial state $|\psi\rangle$.  See Fig. \ref{fig:TrivalentGraph}.

\begin{figure}[t]
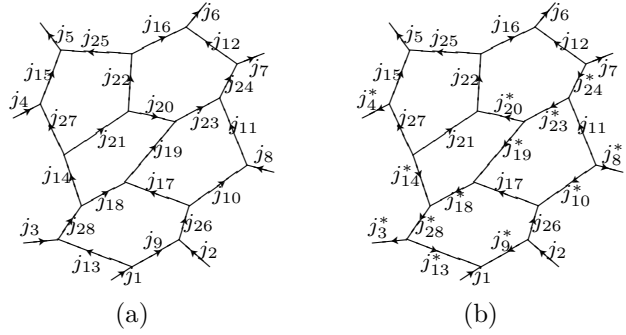

  \centering
  \subfigure[]{\TrivalentGraph}
  \label{fig:TrivalentGraphA}\qquad
  \subfigure[]{\TrivalentGraphII}
  \label{fig:TrivalentGraphB}
  \caption{A configuration of string types on a directed trivalent graph. The configuration (b) is treated the same as (a), with some of the directions of some edges reversed and the corresponding labels $j$ conjugated $j^*$.}
  \label{fig:TrivalentGraph}
 \end{figure}

There are two types of local operators, ${Q}_v$ which are defined at
vertices $v$ and $B_p^s$ which are defined at a plaquette for an $s\in L$. Let
us first define the operator ${Q}_v$.
On a trivalent graph, ${Q}_v$ acts on the labels of three
edges incoming to the vertex $v$. We define the action of ${Q}_v$
on the basis vector with $j_1,j_2,j_3$ by
\begin{align}
  {Q}_v\left|
  \Ypart{j_1}{j_3}{j_2}{>}{<}{<}
  \right\rangle
  =\delta_{j_1j_2j_3}\left|
  \Ypart{j_1}{j_3}{j_2}{>}{<}{<}
  \right\rangle
\end{align}
where the tensor $\delta_{j_1j_2j_3}$ equals either 1 or 0, which
determines whether the triple $(j_1,j_2,j_3)$ is ``allowed'' to
meet at the vertex. Since $\delta_{j_1j_2j_3}=\delta_{j_2j_3j_1}$, the ordering
in the three labels is not important. To be compatible with the conjugation
structure of labels, the branching rule must satisfy
$\delta_{0jj^*}=\delta_{0j^*j}=1$, $\delta_{0ij^*}=0$ if $i\neq{j}$, and
$\delta_{j_1j_2j_3}=\delta_{j_3^*j_2^*j_1^*}$.

One important property of the tetrahedral symbols is that
\begin{equation}
\label{eq:G=Gdelta}
T^{ijm}_{kln}=0 \quad \text{unless } \delta_{ijm}=\delta_{klm^*}=\delta_{lin}=\delta_{nk^*j^*}=1.
\end{equation}
This is a consequence of the orthogonality condition and the tetrahedral symmetry.

For convenience, we take the square root of the loop weight as follows. We define
\begin{equation}
\label{eq:v::Definition}
\rmv_j:=\frac{1}{T^{j^*j0}_{0\,0\,j}}.
\end{equation}
We can verify $\rmv_j^2=\rmw_j$ from the orthogonality condition.
In particular, $\rmv_0=1$.

The operator $B_p^s$ acts on the boundary edges of the plaquette
$p$, and has the matrix elements on a triangle plaquette,
\begin{align}
\label{eq:Bps::InLW}
&\Biggl\langle
\TriangleYpart{j_4}{j^{\prime}_1}{j^{\prime}_2}{j^{\prime}_3}{j_5}{j_6}{>}{<}{>}{<}{}{<}
\Biggr|
B_p^s
\Biggl|\TriangleYpart{j_4}{j_1}{j_2}{j_3}{j_5}{j_6}{>}{<}{>}{<}{}{<}\Biggr\rangle\nonumber\\
=&
\rmv_{j_1}\rmv_{j_2}\rmv_{j_3}\rmv_{j'_1}\rmv_{j'_2}\rmv_{j'_3}
T^{j_5j^*_1j_3}_{sj'_3j^{\prime*}_1}T^{j_4j^*_2j_1}_{sj'_1j^{\prime*}_2}T^{j_6j^*_3j_2}_{sj'_2j^{\prime*}_3}.
\end{align}
The same rule applies
when the plaquette $p$ is a quadrangle, a pentagon,
or a hexagon and so on. Note that the matrix is nondiagonal
only on the labels of the boundary edges (i.e., $j_1$, $j_2$, and
$j_3$ on the above graph).

The operators $B_p^s$ have the properties
\begin{align}
  &B_p^{s\dagger}=B_p^{s^*}
  \label{eq:BpsDagger}\\
  &B_p^rB_{p}^s=\sum_{t}\delta_{rst^*}B_p^t.
  \label{eq:BpsAlgebra}
\end{align}

The Hamiltonian of the model is
\begin{equation}
  \label{eq:HamiltonianLW}
  {H}=-\sum_{v}{Q}_v-\sum_{p}B_p,
  \quad
  B_p=\frac{1}{D}\sum_{s}\rmw_sB_p^{s},
\end{equation}
where $D=\sum_j{\rmd}_j^2$, and the sum runs over all vertices $v$ and all
plaquettes $p$ of the trivalent graph.

The main property of the interactions ${Q}_v$ and $B_p$ is that they are
mutually-commuting, orthogonal projection: (1)
$[Q_v,Q_{v^{\prime}}]=0=[B_p,B_{p^{\prime}}],[Q_v,B_p]=0$; (2)
${Q}_v^2 ={Q}_{v}={Q}_v^*$ and $B_p^2 = B_{p}=B_p^*$.
Thus the Hamiltonian is exactly soluble.
The elementary energy eigenstates are given by common eigenvectors of all these projections.
The ground states have eigenvalues ${Q}_v=B_p=1$
for all $v$ and $p$, while each excited state violates these constraints
for some subset of the plaquettes and vertices.

\subsection{Multi-fusion category extension of the Levin-Wen model}

The input data for LW models can be extended to the multi-fusion case. The
extension is to replace the trivial label 0 by a subset $L_0$ of $L$, in order to numerically
specify the (not necessarily simple) tensor unit of the category.

We start with a \textit{label set} $L$ with an involution $^*:L\rightarrow L$
that is equipped with a trivial set $L_0$, where $L_0$ is determined by the
decomposition of the tensor unit into simple objects as in Sec.
\ref{multi-fusion}.
A \textit{fusion rule} on $L$ is a map $N:L \times L \times L\rightarrow \N$
satisfying that for all $a,b,c,d\in L$,
\begin{align}
&\sum_{\alpha\in L_0}N_{\alpha a}^b=\sum_{\alpha\in L_0}N_{a\alpha}^b=\delta_{ab},
\label{eq:N:a0b:Multifusion}\\
&\sum_{\alpha\in L_0}N_{ab}^{\alpha}=\delta_{ab^*},
\label{eq:ab0:Multifusion}\\
&\sum_{x\in L}N_{ab}^xN_{xc}^d=\sum_{x\in L}N_{ax}^dN_{cd}^x.\label{eq:NN=NN:Multifusion}
\end{align}
These three equations are obtained by formally replacing 0 by $\sum_{\alpha\in L_0}\alpha$ in Eqs. \eqref{eq:N:a0b}, \eqref{eq:N:ab0} and \eqref{eq:NN=NN}. Since $N_{a\alpha}^b\in\N$, the first equality implies that for each label $a\in L$, there exists a unique pair $(\alpha,\beta)\in L_0\times L_0$ such that $N_{\alpha' a}^b=\delta_{ab}\delta_{\alpha'\alpha}$ and $N_{a \beta'}^b=\delta_{ab}\delta_{\beta'\beta}$ for $b\in L, \alpha',\beta'\in L_0$. We say $a$ has the grading $(\alpha,\beta)$. Obviously, each $\alpha\in L_0$ has the grading $(\alpha,\alpha)$.

Therefore $L$ is graded by $L_0\times L_0$: $L=\underset{\alpha,\beta\in L_0}{\sqcup} {}_{\alpha}L_{\beta}$, and we can denote the labels in ${}_{\alpha}L_{\beta}$ by ${}_{\alpha}a_{\beta}$ to specify their gradings $(\alpha,\beta)$. Eqs. \eqref{eq:N:a0b:Multifusion} and \eqref{eq:NN=NN:Multifusion} imply
\begin{equation}
N_{_\alpha a_\beta,_\gamma b_\delta}^{_\epsilon c_\zeta}=0\quad \text{ unless } \alpha=\epsilon,\beta=\gamma,\delta=\zeta.
\end{equation}
Together with Eq. \eqref{eq:ab0:Multifusion}, it implies
\begin{align}
&\alpha^*=\alpha \quad \text{for } \alpha\in L_0,\\
&{_\alpha a_\beta}^*\in {_\beta}L_{\alpha} \quad \text{for } {_\alpha a_\beta}\in {_\alpha L_\beta}.
\end{align}

Given a fusion rule on $\{L,L_0\}$, the loop weight satisfies
\begin{equation}
\sum_{_\alpha c_{\gamma}\in _\alpha L_{\gamma}}\delta_{_\alpha a_\beta,{}_\beta b_{\gamma},(_\alpha c_{\gamma})^*}\rmw_{_\alpha c_\gamma}=\rmw_{_\alpha a_\beta}\rmw_{_\beta b_\gamma}.
\end{equation}
The symmetrized tetrahedral symbols are defined in the same way as those in the previous section, and so are the LW models. This leads to the following conclusion:

\begin{prop}

Using the modified label set $L$ with trivial set $L_0$, the LW Hamitonian
schemas extend to multi-fusion categories, and all resulting Hamiltonians are CLPs.

\end{prop}

\subsection{The $n\times n$ $\two$-matrix $\cM_n$ as input}\label{sec:Example:SimplestCase}

Consider the multi-fusion category $\cM_n$ from example~\ref{example:matrices}. This example gives the following data. The label set is $L=\{E_{ij}\}$, the trivial set is $L_0=\{E_{ii}\}$, and the fusion rule is
\begin{equation} \delta_{E_{ij},E_{kl},E_{mn}}=\delta_{jk}\delta_{lm}\delta_{ni}.
\label{eq:FusionRule:MatrixAlgebra}
\end{equation}
The set $L=\underset{i,j}{\sqcup}\,{{}_i L_j}$ is graded by $i,j$ where each $_i
L_j$ has only one element, $E_{ij}$. The duals are $E_{ij}^*=E_{ji}$.

Let us set the loop weights to be $\rmw_{E_{ij}}=1$ for all $i,j$. The simplest normalized $6j$-symbol is to take
\begin{equation}
\label{eq:G=Gdelta:Multifusion}
T^{abc}_{def}=
\left \{
\begin{array}{ll}
 1 & \text{if } \delta_{abc}=\delta_{dec^*}=\delta_{eaf}=\delta_{fd^*b^*}=1, \\
 0 & \text{otherwise.} \\
\end{array} \right.
\end{equation}
for $a,b,c,d,e,f\in L$.

The local Hilbert space is spanned by labels on all edges. In our example,
labels are the gradings $(i,j)$. Graphically, we use a double line to represent
the gradings as illustrated below.
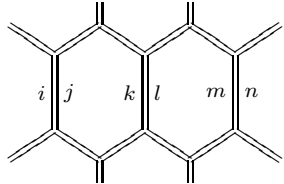
\[
\DoubleLineHilbertSpace
\]
We do not draw arrows in the graph as a label on each arrowed edge is identified
with its dual on the same edge with the arrow reversed. For example, the labels on the three vertical edges illustrated above read as $E_{ij}$, $E_{kl}$ and $E_{mn}$ upwards, and as $E_{ji}$, $E_{lk}$ and $E_{nm}$ downwards.

Consider the eigenspace $\scH^{Q=1}$ of $Q_v=1$ for all vertices. The fusion
rule in Eq. \eqref{eq:FusionRule:MatrixAlgebra} has a double line representation
near each vertex of the form
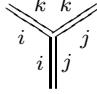
\[
\DoubleLineVertex
\]
which presents an admissible triple $(E_{ij},E_{jk},E_{kl})$ on the three edges
incoming into the vertex, and for which all other combinations are not allowed.
If two lines are connected, then they carry the same label $i$.

Therefore the basis vectors in $\scH^{Q=1}=\otimes_{p}\C^n$ have a double line
representation as below.
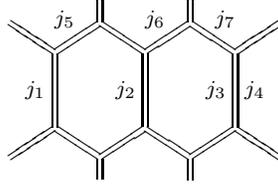
\[
\DoubleLineSpaceA
\]
To each plaquette $p$, there is a loop labeled by $j_p$. The basis is denoted in
terms of the loop labels $j_p$ and given by $\{\ket{j_1,j_2,\dots}\}$. This
statement holds for the model on any closed surface.

The operator $B_p$ is now $B_p=\frac{1}{n}\sum_{\alpha\beta}B_p^{E_{\alpha\beta}}$, where $B_p^{E_{\alpha\beta}}$ is defined in Eq. \eqref{eq:Bps::InLW}. In the subspace $\scH^{Q=1}$, $B_p^{E_{\alpha\beta}}$ is a map
\begin{equation}
B_p^{E_{\alpha\beta}}:
\ket{j_1,j_2,\dots,j_p,\dots}\mapsto
\delta_{\beta,j_p}\ket{j_1,j_2,\dots,\alpha,\dots}.
\end{equation}

Therefore there is only one ground state, with common eigenvalues $Q_v=1$ and $B_p=1$ for all $v,p$:
\begin{equation}
\ket{\Phi}=\sum_{\alpha_1,\alpha_2,\dots}\ket{\alpha_1,\alpha_2,\dots,\alpha_p,\dots},
\label{eq:GroundState:Phi}
\end{equation}
up to a constant normalization factor. The discussion can be summarized by the following proposition.

\begin{prop}\label{trivialTQFT}

The LW Hamiltonian schemas with input $\cM_n$ for all $n\geq 1$ realize the trivial $(2+1)$-TQFT.

\end{prop}

Consider now the example $n=2$, for which it is easy to give an explicit
description of the ground state. In this case the operator $B_p$ is the matrix
$\frac{1}{2}(\dsI+\sigma^x)$ in the local basis $\ket{i_p}$, where
$\sigma^x={\mat[0.8]{0&1\\1&0}}$ is a Pauli matrix. Dropping the constant terms,
we can write the Hamiltonian in the subspace $\scH^{Q=1}$ as
\begin{equation}
H|_{Q = 1}=-\frac{1}{2}\sum_p \sigma^x_p.
\label{eq:Hamiltonian:sigmaX}
\end{equation}
It is convenient to use the dual graph picture. Namely, by
taking the dual graph of a spatial trivalent graph, we obtain a
triangulation of the surface.  Then the ground state is simply a tensor product $\otimes_p\ket{\sigma^x_p=1}$ of all local
eigenstates of $\sigma^x=1$ at the vertices of the dual triangulation.

\subsection{Degeneracy on a Disk}

Consider the disk with a smooth loop boundary. On the graph in Fig.
\ref{fig:DiskWithBoundary}(a), the Hamiltonian takes the form in Eq.
\eqref{eq:HamiltonianLW}, with the first summation over all vertices of the
graph and over all internal plaquettes inside the disk.

\begin{figure}[h]
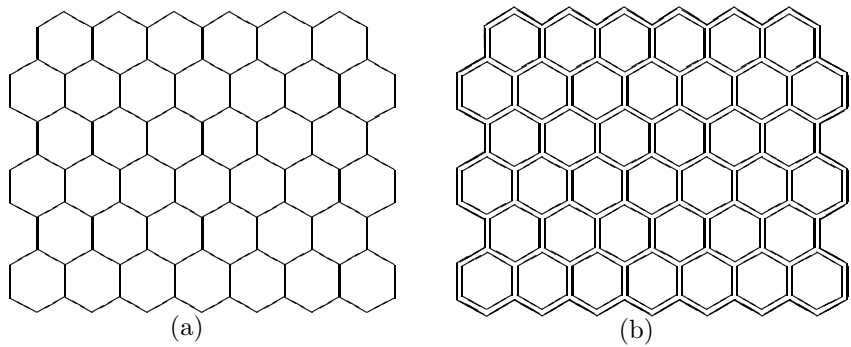

	\centering
	\subfigure[]{\DiskWithBoundary}
	\quad
	\subfigure[]{\DoubleLineDisk}
	\caption{(a). Disk with a loop boundary. (b). Double line representation for $\scH^{Q=1}$.}
	\label{fig:DiskWithBoundary}
\end{figure}

The double line representation for $\scH^{Q=1}$ is illustrated in Fig.
\ref{fig:DiskWithBoundary}(b).  A basis vector in $\scH^{Q=1}$ is denoted by
$\ket{\alpha_{\partial};\alpha_1,\alpha_2,\dots,\alpha_p,\dots}$, specified by a
loop value $\alpha_p$ associated to each plaquette $p$ inside the disk, and a
loop value $\alpha_{\partial}$ associated to the boundary.

The second term $-\sum_p B_p$ in the Hamiltonian does not affect
$\alpha_{\partial}$. Therefore, the ground states are degenerate and
paramterized by $\alpha_{\partial}$. For the input data $\cM_n$, the ground
state degeneracy is $n$.

Similar to the formula in Eq.\eqref{eq:GroundState:Phi}, the degenerate ground
states for all $\alpha_{\partial}$ are
\begin{equation}
\ket{\Phi(\alpha_{\partial})}=\sum_{\alpha_1,\alpha_2,\dots}\ket{\alpha_{\partial};\alpha_1,\alpha_2,\dots,\alpha_p,\dots}
\label{eq:GroundStateDisk}
\end{equation}

\subsection{Topological Entanglement Entropy}

Consider the extended LW model with $\cM_n$ as input.  We divide a trivalent graph into two subsystems $A$ and $B$, where their boundary intersects some edges, denoted by a dashed curve as illustrated in Fig. \ref{fig:Bipartition}.

\begin{figure}[h]
\centering
\[\Bipartition\]
\caption{Partition into subsystems $A$ and $B$ with the boundary along a dashed curve.}
\label{fig:Bipartition}
\end{figure}

Denote the edges across the boundary by $j_1,j_2,\dots,j_l\in L$,
or simply $\{j_i\}$ for short. The number $l$ will be called the length of the boundary
curve.

The reduced density matrix for the ground state $\Phi$ in Eq.~(\ref{eq:GroundState:Phi}) is defined by $\rho_A=\oplus_{\{j_i\}}\rho_A^{\{j_i\}}$, where
\begin{equation}
\rho_A^{\{j_i\}}=\mathrm{tr}_B[\bra{\{j_i\}}(\ket{\Phi}\bra{\Phi})\ket{\{j_i\}}].
\end{equation}
Here $\mathrm{tr}_B$ is the partial trace over all labels in the subsystem $B$.

By definition, the entanglement entropy is
\begin{equation}
S_E=-\mathrm{tr}_A(\rho_A \mathrm{log}\rho_A),
\end{equation}
where we calculate the entanglement entropy on the $2$-sphere.

\begin{figure}[h]
	\centering
	\[\DoubleLineGroundState\]
	\caption{Nonzero contributions to the entanglement spectrum are
	specified by the loop labels $\alpha_1,\alpha_2,\dots,\alpha_l$ on
	the boundary.}
	\label{fig:DoubleLineGroundState}
\end{figure}

The double line representation provides a clear picture of the spectrum of
$\rho_A$: Nonzero contributions to the entanglement spectrum are specified by
the loop labels $\alpha_1,\alpha_2,\dots,\alpha_l$ on the boundary, see Fig.
\ref{fig:DoubleLineGroundState}. Specifically, in terms of the new basis of the
subspace $\scH^{Q=1}$, the boundary is specified by the loop labels
$\alpha_1,\alpha_2,\dots,\alpha_l$. $\rho_A^{\{j_i\}}$ has exactly
one nonzero eigenvalue $\lambda$ if and only if the boundary configuration $\{j_i\}$ has the
following form:

\[
\xy
{\aline^{\alpha_l} (0,0);(0,10)};
{\aline_{\alpha_1} (2,0);(2,10)};
{\aline^{\alpha_1} (20,0);(20,10)};
{\aline_{\alpha_2} (22,0);(22,10)};
{\aline^{\alpha_2} (40,0);(40,10)};
{\aline_{\alpha_3} (42,0);(42,10)};
(60,5)*{\dots};
{\aline^{\alpha_{l-1}} (80,0);(80,10)};
{\aline_{\alpha_l} (82,0);(82,10)};
\endxy
\]

By symmetry, $\rho_A$ has $n^l$ equal eigenvalues, which are
normalized to $\lambda=1/n^l$ by the trace condition $\mathrm{tr}_A(\rho_A)=1$. It follows
that
\begin{equation}
S_E=\mathrm{log}(n)l.
\end{equation}
Since there is not any sub-leading correction term in $S_E$ --- it is exactly
proportional to the length $l$ of the
boundary curve --- the topological entanglement entropy is $0$~\cite{KP,LW2}.  A
similar calculation on the torus also leads to zero topological entanglement entropy.

\section{Symmetry Enriching the Levin-Wen model}

We are interested in enriching the LW model with on-site unitary symmetries.  A
good example is the toric code Hamiltonian $H=-\sum_{v}A_v-\sum_{p}B_p$ on the
square lattice, where a qubit is one each edge.  As usual, the vertex operator
$A_v$ is the tensor product of $\sigma^x$ and the identity, while the plaquette
term is a tensor product of $\sigma^z$ and the identity.  A moment's thought shows
that the tensor product of $\sigma^x$ (or $\sigma^z$) over all edges is an on-site unitary symmetry of the toric code Hamiltonian.  Of
course this $\Z_2$ symmetry is very trivial because it will not permute anyon
types.  But even if a $\Z_2$ symmetry of the toric code does not permute anyon
types, there are still four different ways to fractionalize a $\Z_2$ symmetry in
a one-to-one correspondence to classes in $H^2(\Z_2;\Z_2^2)=\Z_2^2$ \cite{BBCW}.
In this section, we will describe analogous symmetries of the LW Hamiltonians.  It
will be interesting to understand their role in a microscopic theory of symmetry
fractionalization, symmetry defects, and gauging using fixed-point rigorously
solvable Hamiltonians.

\subsection{Classification of $n\times n$ $\two$-matrices}

The half-label set can be endowed with a group structure.  In this subsection, we
classify all $n\times n$ $\two$-matrices whose half-label set has the structure of an abelian group $G$.

By the fusion rule, there are four independent variables in the $6j$-symbols. Denote them by
\begin{equation}\label{G6j}
\phi_4(\alpha,\beta,\gamma,\delta)
:=
T^{E_{\alpha\beta}E_{\beta\gamma}E_{\gamma\delta}}_{E_{\gamma\delta}E_{\delta\alpha}E_{\beta\delta}}\rmw_{E_{\beta\delta}}.
\end{equation}
In this notation the pentagon identity can be written as
\begin{equation}
\phi_4(\alpha,\beta,\gamma,\delta)
\phi_4(\alpha,\beta,\delta,\epsilon)
\phi_4(\beta,\gamma,\delta,\epsilon)
=
\phi_4(\alpha,\gamma,\delta,\epsilon)
\phi_4(\alpha,\beta,\gamma,\epsilon),
\label{eq:Pentagon:phi}
\end{equation}
for $\alpha,\beta,\gamma,\delta=1,2,\dots, n$.

Suppose the half labels $\alpha,\beta,\dots$ form a finite group $G$ with
$|G|=n$, e.g. $G=\Z_n$. Recall that a homogeneous $n$-cochain taking values in
$\mathbb{C}$ is a map $\phi_{n+1}: G^{n+1} \to \mathbb{C} \setminus \{0\}$ such
that $g \cdot \phi_{n+1}(g_1, \dots g_{n+1}) = \phi_{n+1}(gg_1, \dots,
gg_{n+1})$. We will usually consider the trivial $G$-action on $\mathbb{C}
\setminus \{0\}$. Hence, $\phi_4:G^4\rightarrow \dsC\backslash\{0\}$ is a
homogeneous 3-cochain on $G$, equipped with an action:
\begin{equation}
g\cdot\phi_4(\alpha,\beta,\gamma,\delta)
=\phi_4(g\alpha,g\beta,g\gamma,g\delta),
\label{eq:3cocycle:action}
\end{equation}
where we regard $\mathbb{C} \setminus \{ 0 \}$ as a trivial $G$-module. The pentagon identity \eqref{eq:Pentagon:phi} can then identified with the 3-cocycle condition $\delta\phi_4=1$, where the coboundary $\delta$ is defined by
\begin{equation}
\delta\phi_4(\alpha_0,\alpha_1,\dots,\alpha_4)=
\prod_{0\leq i\leq 4}\phi_4(\alpha_0,\alpha_1,\dots,\alpha_{i-1},\alpha_{i+1},\dots,\alpha_4)^{(-1)^i}.
\end{equation}
Therefore, the $6j$-symbols are classified by the third group cohomology classes
in $H^3(G,U(1))$.  Note that not all $3$-cocycles satisfy the tetrahedral
symmetry in Eq. \eqref{eq:6j:TetSymmetry}. We call $3$-cocycles $\phi_4$ defined
as above $G$-invariant.

\begin{definition}

Given a finite group $G$ and a homogeneous $3$-cocycle $\phi_4$, $\phi_4$ is
called \textit{$G$-invariant} if $\phi_4(\alpha,\beta,\gamma,\delta)
=\phi_4(g\alpha,g\beta,g\gamma,g\delta)$ for all
$\alpha,\beta,\gamma,\delta=1,\cdots,n,$ and $ g \in G$. I.e. the action of $G$
on $\phi_4$ given by Eq. \eqref{eq:3cocycle:action} is trivial if
$\dsC\backslash\{0\}$ is regarded as a trivial $G$-module.

\end{definition}

Consider the case where $n=2$. Then the group is $\Z_2=\{0,1\}$. There are two equivalence classes, with the 3-cocycle representatives:
\begin{enumerate}
\item $\rmw_{E_{\alpha\beta}}=1$, and $\phi_4=1$ is constant, as in Sec. \ref{sec:Example:SimplestCase};
\item $\rmw_{E_{\alpha\beta}}=
\left\{
\begin{array}{ll}
 1 & \text{if } \alpha=\beta \\
 -1 & \text{if }\alpha \ne \beta \\
\end{array} \right.$, and
\[\phi_4(\alpha,\beta,\gamma,\delta)=\exp\left[\frac{\pi i}{2}(2-|\alpha+\beta+\gamma+\delta-2|)\right]\rmw_{E_{\beta\delta}.}\]
\end{enumerate}
The two representatives are chosen to satisfy the tetrahedral symmetry in Eq. \eqref{eq:6j:TetSymmetry}. The $G$-actions in Eq. \eqref{eq:3cocycle:action} on both $3$-cocycles are trivial, hence both $3$-cocycles are $\Z_2$-invariant.

Similar to Eq. \eqref{eq:Hamiltonian:sigmaX}, the Hamiltonian for the second class can be written as
\begin{equation}
H=-\frac{1}{2}\sum_p \tau^x_p.
\label{eq:Hamiltonian:TiledSigmaX}
\end{equation}
In the dual triangulation, $\tau^x$ is
\begin{equation}
\tau^x=\left\{\prod_{\langle ij\rangle\in\partial p}\exp\left[i\frac{\pi}{4}(\dsI-\sigma^z_i\sigma^z_{j})+i\frac{\pi}{2}(\dsI+\sigma^z_i\sigma^z_{j})\right]\right\}\sigma^x_p,
\end{equation}
with the product over nearest neighbor vertex pairs on the boundary of $p$,
for example, over $\langle12\rangle,\langle23\rangle,\dots,\langle61\rangle$ in
the example below:
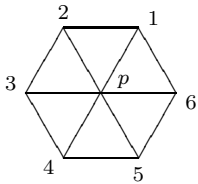
\[\TriangleBp\]
Here only the relevant triangles of the dual graph are shown, assuming the remaining part of the graph is not affected.

\subsection{$G$-symmetric Hamiltonian Schema}

Given a homogeneous $3$-cocycle $\phi_4$, not necessarily $G$-invariant, we have
a multi-fusion category $(\cM_n,\phi_4)$ with $6j$-symbols given by
Eq. \eqref{G6j}. This in turn allows us to define a Levin-Wen Hamiltonian schema
with this multi-fusion category as input.

\begin{definition}
Given a finite group $G$ and a Levin-Wen Hamiltonian schema,  the Levin-Wen
Hamiltonian schema is \textit{$G$-symmetric} if each $g\in G$ acts on the qudit
$\C^d$ as a unitary matrix $U_g$, such that it is a symmetry of all resulting Levin-Wen
Hamiltonians.
\end{definition}

\begin{theorem}\label{Thm:SPT}

If the homogeneous $3$-cocycle $\phi_4$ for an $n\times n$ $\two$-matrix is $G$-invariant, then the Levin-Wen Hamiltonian schema with the $n\times n$ $\two$-matrix $(\cM_n,
\phi_4)$ input is $G$-symmetric, and realizes a $G$-symmetry protected topological phase ({\textrm{SPT}}).

\end{theorem}

Using Prop. \ref{trivialTQFT}, we just need to check the $G$-invariance of Levin-Wen Hamiltonians, which is a straightforward check.  
But it is not clear if we have realized any non-trivial SPTs, which will be addressed in the next section.

We conjecture that this result can be extended in the following way.
\begin{conj}

The LW Hamiltonian schema with an $n\times n$ multi-fusion $\mcc$ input realizes a symmetry enriched topological phase $D(\mcc)$ with some on-site unitary symmetry $G$, which does not permute anyon types.

\end{conj}

\subsection{De-equivariantizing the $G$-symmetric Levin-Wen model}
To understand if the SPTs realized in Thm.\ref{Thm:SPT} are non-trivial, we study the gauging of the symmetry $G$ \cite{LG,BBCW}.
First we give a proof of the following proposition.

\begin{prop}

There is a non-local transformation from $G$-symmetric LW models to traditional LW models coupled to a local action.

\end{prop}

Given a finite group $G$, a homogeneous $3$-cocycle $\phi_4$ of $G$ can be de-equivariantized to obtain an inhomogeneous $3$-cocycle $\varphi_3$ by setting
\begin{equation}
\varphi_3(x,y,z)=\phi_4(1,x,xy,xyz),
\end{equation}
for $x,y,z\in G$ and $1$ is the identity element of $G$.   The $3$-cocycle $\varphi_3$ has a group action
\begin{equation}
g\cdot \varphi_3(x,y,z)=\phi_4(g,gx,gxy,gxyz).
\end{equation}
The inhomogeneous $3$ cocycles $\varphi_3$ and homogeneous $3$-cocycles $\phi_4$ are in one-one correspondence because $\phi_4$ can be recovered from $\varphi_3$ by
\begin{equation}
\phi_4(\alpha,\beta,\gamma,\delta)=\alpha\cdot \varphi_3(\alpha^{-1}\beta,\beta^{-1}\gamma,\gamma^{-1}\delta).
\label{eq:phi4tophi3}
\end{equation}

This de-equivariantization reduces the $G$-symmetric data from a multi-fusion category to input data from an abelian modular category $\mathcal{V}ec_G^{\varphi_3}$ with a nontrivial action of $G$ on $\varphi_3$.

The correspondence between $\phi_4$ and $\varphi_3$ can be adapted to the local Hilbert spaces and their Hamiltonians, therefore, the correspondence establishes a non-local duality transformation. In the following, we will work with the dual triangulations and consider only the $2$-sphere $S^2$ for simplicity.

For the local Hilbert spaces, the subspaces $\scH^{Q=1}$ are spanned by the group elements $\{\alpha_p\}$ at vertices $p$ of the dual triangulations. Choose an arbitrary vertex $p_0$, and designate it as the {\it origin}.

On the $2$-sphere, the set of group elements $\{\alpha_0,\alpha_1,\alpha_2,\dots\}$ assigned to vertices corresponds to the set of group elements $\{g_1,g_2,\dots\}$ assigned to edges satisfying the following condition: around any triangle, the holonomy (the product of the three group elements around the triangle) is equal to the identity $1$. In fact, the group element $g_e$ on each edge $e$ can be written as $g_e=\alpha_2\alpha_1^{-1}$, so it is determined by $\alpha_1$ ($\alpha_2$) at the starting (ending) point of $e$. Conversely, given $\alpha_0$ at the origin vertex $p_0$, $\alpha_p$ can be determined as follows: choose an arbitrary path from $p_0$ to $p$, multiply the group elements on the edges along the path and $\alpha_0$. The two constructions above give rise to an isomorphism
\begin{equation}
\{\alpha_0,\alpha_1,\alpha_2,\dots\}|_{\text{vertex colors}}\cong
\{\alpha_0;g_1,g_2,\dots\}|_{\text{trivial holonomy}}.
\label{eq:isomorphism::Sphere}
\end{equation}
where ``trivial holonomy'' means that the group elements $g$ around each triangle have a product equal to the identity $1$. Therefore, the Hilbert space $\scH^{Q=1}$ has a basis
\begin{equation}
\{\ket{\alpha_0;g_1,g_2,\dots}\}\left|_{\text{trivial holonomy}}\right.
\label{eq:basis:agg}
\end{equation}

If the $G$-action is trivial, then the $G$-symmetric Hamiltonian can be
de-e\-qui\-var\-i\-ant\-ized as follows. First, $\varphi_3$ produces new input data
$\{\tilde{\rmw},\tilde{\delta},\tilde{T}\}$, where $g,g_1,g_2,g_3\in G$, by
defining
\begin{align}
&\tilde{\rmw}_{g}=\rmw_{E_{1g}},\label{eq:TraditionalCoupledToG12}\\
&\tilde{\delta}_{g_1,g_2,g_3}=\delta_{g_1g_2g_3,1},\\
&\tilde{T}^{g_1,g_2,(g_1g_2)^{-1}}_{g_3,(g_1g_2g_3)^{-1},g_2g_3}=\varphi_3(g_1,g_2,g_3)/\rmw_{g_2g_3}.
\label{eq:TraditionalCoupledToG14}
\end{align}
Then, the Hamiltonian in terms of $\{\tilde{\rmw},\tilde{\delta},\tilde{T}\}$ is
\begin{equation}
H=-\sum_{v}\tilde{Q}_v-\sum_{p}\tilde{B}_p,
\label{eq:Hamiltonian::TraditionalCoupledToG}
\end{equation}
where $\tilde{B}_p=\frac{1}{n}\sum_{g}\rmw_g\tilde{B}_p^g$ for all plaquettes except for $p_0$, and $\tilde{B}_p^g$ is defined as in Eq. \eqref{eq:Bps::InLW} in terms of $\varphi_3$, which acts on the degrees of freedom $g_1,g_2,\dots$ in the basis \eqref{eq:basis:agg}.

At $p_0$, $\tilde{B}_{p_0}=\frac{1}{n}\sum_{g}\rmw_g\tilde{B}_{p_0}^g T^g_{p_0}$, where
\begin{equation}
T^g_{p_0}: \ket{\alpha_0;g_1,g_2,\dots}\mapsto
\ket{g\alpha_0;g_1,g_2,\dots}.
\end{equation}

Therefore, the non-local transformation defines a one-to-one correspondence between the $G$-symmetric LW models and the modified traditional LW models with input data from $\mathcal{V}ec_G^{\varphi_3}$ and $\tilde{B}_{p_0}$ coupled to the local group action $T^g_{p_0}$. The local group action $T^g_{p_0}$ corresponds to a global action in the $G$-symmetric LW model:
\begin{equation}
T^g:\ket{\alpha_0,\alpha_1,\dots,\alpha_p,\dots}
\mapsto
\ket{g\alpha_0,g\alpha_1,\dots,g\alpha_p,\dots}
\label{eq:GlobalAction}
\end{equation}

Let us apply the non-local transformation on the ground state $\Phi$ on the $2$-sphere.
In the transformed traditional LW model, the ground state is the common eigenstate of $\tilde{B}_p^g=1$, for $p\ne p_0$, and $\tilde{B}_{p_0}^gT_{p_0}^g=1$, for all $g\in G$. The global constraint in the traditional LW model enforces $\tilde{B}_{p_0}^g=1$ and hence $T_{p_0}^g=1$. By the non-local transformation, $T_{p_0}^g=1$ means that the ground state is invariant under the global symmetry $\{T^g\}$ in the $G$-symmetric LW model.

{\bf Physical Theorem}\footnote{By a physical theorem, we mean that the argument is only rigorous physically. Therefore, physical theorems should be regarded as mathematical conjectures.}: {\it The $G$-symmetric LW model with input $\cM_n$ realizes a $G$-{\textrm{SPT}} with the $3$-cocycle $\varphi_3 \in H^3(G;U(1))$ when $\varphi_3$ is $G$-invariant}.

We did not prove this theorem mathematically because we did not define universality classes of SPT phases mathematically.  But physically we summarize the argument above as follows.
Each $G$-invariant $3$-cocycle $\varphi_3$ leads to an SPT because the LW model realizes the trivial TQFT.  To understand the local term $T_{p_0}^g$, we map the SPT model to a nontrivial TQFT coupled to a gauge field with a gauge coupling term, where the half-labels represent the gauge field. If we eliminate the gauge coupling term, all half-labels are eliminated as well except the one at the base point. This leaves behind the local term at the base point.

\begin{remark}The input $6j$-symbols in Eqs. \eqref{eq:TraditionalCoupledToG12}-\eqref{eq:TraditionalCoupledToG14} are well-defined only when the G-action on $\varphi_3$ is trivial. So de-equivariantization works only for trivial $G$-actions. If the $G$-action on $\varphi_3$ is nontrivial, then the $6j$-symbols are equipped with a G-action, which leads to a LW model with a gauge group action.
\end{remark}

\subsection{On a Disk}

Consider further a disk with a smooth boundary, e.g., with the graph in Fig. \ref{fig:DiskWithBoundary}(a). The non-local transformation leads to the same form of the Hamiltonian as in Eq. \eqref{eq:Hamiltonian::TraditionalCoupledToG}, but with the second summation over all plaquettes $p$ inside the disk. The degenerate ground states $\Phi(\alpha_{\partial})$ in the $G$-symmetric LW model are parameterized by the half-label $\alpha_{\partial}$. Now let us reexamine the ground states in the traditional LW model under the non-local transformation.

Take an arbitrary plaquette inside the disk as the origin, denoted by $p_0$. The
ground states are the common eigenstates of $\tilde{B}_p^g=1$, for $p\ne p_0$
inside the disk, and $\tilde{B}_{p_0}^gT_{p_0}^g=1$, for all $g\in G$. Due to
the presence of the boundary, the global constraint on $\tilde{B}_{p_0}^g$ is
released. If $\tilde{B}_{p_0}^g$ transforms under a non-trivial irreducible representation
$\rho$ of $G$, we say there is an elementary quasiparticle (or a topological
defect) at $p_0$ identified by its topological charge $\rho$. This topological
charge is always coupled to a charge which transforms under the dual
representation $\rho^*$ of the local group action.

The degenerate ground states $\Phi_{\rho}$ are thus parametrized by the charge
$\rho$. Under the non-local transformation, they correspond to the ground states
in the $G$-symmetric LW model, carrying a global charge $\rho^*$ under the
global symmetry $\{T^g\}$. Meanwhile, the topological charge $\rho$ of the local
quasiparticle in the traditional LW model is mapped to the boundary condition
specified by $\rho$ in the $G$-symmetric LW model. This relation between
$G$-symmetric LW models and LW models coupled to a gauge action is listed in
Table~\ref{tab:Disk}.

\begin{table}[h]
\caption{Non-local transformation on a disk.}
\label{tab:Disk}
\begin{center}
	\begin{tabular}{c|c}
		\hline
		\hline
		$G$-symmetric LW model & Traditional LW model coupled to a local action\\
		\hline\hline
		global symmetry & a local action on Hamiltonians\\
		\hline
		boundary condition  &
		bulk local quasiparticle\\
		specified by $\rho$ & with topological charge $\rho$\\
		\hline
		global charge $\rho^*$ &
		a local charge $\rho^*$ coupled to the quasiparticle\\
		\hline
	\end{tabular}
\end{center}
\end{table}

For example, take $\cM_n$ as the input data, and let $G=\Z_n$. The
degenerate ground states can be parameterized by the charge $k=0,1,\dots,n-1$ of $\Z_n$, being the eigenvectors of
\begin{equation}
\tilde{B}_{p_0}^g=\exp\left(\frac{2k\pi g i}{n}\right),
T_{p_0}^g=\exp\left(-\frac{2k\pi g i}{n}\right).
\end{equation}
Such ground states $\Phi_k$ are related to $\Phi(\alpha_{\partial})$ by the following Fourier transformation
\begin{equation}
\Phi_k=\frac{1}{\sqrt{n}}\sum_{\alpha_{\partial}}\exp\left(\frac{2k\pi \alpha_{\partial} i}{n}\right)\Phi(\alpha_{\partial}).
\end{equation}
One can verify the identity by applying the action of $T^g$ in Eq. \eqref{eq:GlobalAction} directly.

\subsection{On a General Closed Surface} The de-equivariantization can be applied on an arbitrary closed surface $Y$ in a similar way. The isomorphism in Eq. \eqref{eq:isomorphism::Sphere} is replaced by
\begin{equation}
\{\alpha_0,\alpha_1,\alpha_2,\dots\}|_{\text{vertex colors}}\cong
\{\alpha_0;g_1,g_2,\dots\}|_{\text{trivial homotopy \& trivial holonomy}},
\label{eq:isomorphism::RiemannianSurface}
\end{equation}
where trivial homotopy means that along any non-contractible loop on the
dual-triangulation of the graph, the group elements $g$ multiply to the
identity element of $G$.

$G$-symmetric LW models are transformed to traditional LW models in the trivial homotopic Hilbert subspace coupled to a local action. The models are well defined because the Hamiltonian is invariant in the trivial homotopic Hilbert subspace.

\section{Open Questions}
We have studied how Levin-Wen models can be extended to take multi-fusion categories as their input, and how on-site symmetries play a role. There are however still interesting open questions. We mention a few:

\begin{enumerate}

\item Classify $n\times n$ $\two$-matrices.

\item Prove that the LW model with an indecomposable multi-fusion category input $\mcc=\oplus_{ij}\mcc_{ij}$ realizes the Turaev-Viro TQFT based on $\mcc_{ii}$ for some $i$.

\item How to realize symmetry fractionalization, symmetry defects, and gauging
with LW models.

\end{enumerate}

\end{document}